\documentclass[11pt,a4paper]{article}
\usepackage{indentfirst,mathrsfs}
\usepackage{amsfonts,amsmath,amssymb,amsthm}
\usepackage{latexsym,amscd}
\usepackage{amsbsy}
\usepackage[all]{xy}
\usepackage{graphicx}
\usepackage{threeparttable}
\usepackage[usenames,dvipsnames,svgnames,table]{xcolor}
\usepackage{cite}

\newtheorem{thm}{Theorem}
\newtheorem{lem}{Lemma}
\newtheorem{cor}{Corollary}

\newtheorem{prop}{Proposition}

\newtheorem{example}{Example}
\newtheorem{remark}{Remark}

\begin{document}

\title{Symbol-pair Weight Distributions of Some Linear Codes}
\maketitle
\begin{center}\author{\large Junru Ma \quad\quad Jinquan Luo\footnote{Corresponding author
\par
 The authors are with School of Mathematics
and Statistics \& Hubei Key Laboratory of Mathematical Sciences, Central China Normal University, Wuhan China, 430079.
\par  E-mails: junruma@mails.ccnu.edu.cn(J.Ma), luojinquan@mail.ccnu.edu.cn(J.Luo).
}}\end{center}

\begin{quote}
{\small {\bf Abstract:} \ \
In this paper, we investigate the symbol-pair weight distributions of MDS codes and simplex codes over finite fields.
Furthermore, the result shows that all the nonzero codewords of simplex codes have the same symbol $b$-weight and rearrangement entries of codewords in simplex code may induce different symbol-pair weight.
In addition, the symbol $b$-weight distributions of variation simplex codes over certain finite fields are derived.
}


\end{quote}

\section{Introduction}

Symbol-pair codes were first proposed in \cite{CB1} by Cassuto and Blaum.
There exists symbol-pair codes with rates that are strictly higher than the best known codes in the Hamming metric \cite{CB2, CL}.
Chee et al. established a Singleton-type bound on symbol-pair codes and constructed many classes of MDS symbol-pair codes using the known MDS codes and interleaving technique in \cite{CJKWY,CKW}.
Ding et al. extended the Singleton-type bound to the $b$-symbol case, and constructed MDS $b$-symbol codes based on projective geometry \cite{DZG}.
Moreover, several works are contributed to the constructions of symbol-pair codes meeting the Singleton-type bound \cite{CLL,DGZZZ,KZL,KZZLC,LG}.
Kai et al. constructed MDS symbol-pair codes from
{contacyclic} codes \cite{KZL}.
Chen et al. established MDS symbol-pair codes of length $3p$ through repeated-root cyclic codes \cite{CLL}.
Li and Ge provided a number of MDS symbol-pair codes by analyzing certain linear fractional transformations \cite{LG}.
{Ding et al.}
obtained some MDS symbol-pair codes utilizing elliptic curves \cite{DGZZZ}.
Three new families of MDS symbol-pair codes were constructed based on repeated-root codes in \cite{KZZLC}.
Meanwhile, some decoding algorithms of symbol-pair codes were proposed by various researchers.
Yaakobi et al. proposed efficient decoding algorithms for cyclic symbol-pair codes in \cite{YBS2,YBS1}.
A decoding algorithm utilizing the syndrome of symbol-pair codes was provided in \cite{HTM}.
An error-trapping decoding algorithm required to impose some restrictions on the symbol-pair error patterns subsequently was shown in \cite{MHT}.
The decoding algorithm based on linear programming was designed for binary linear symbol-pair codes in \cite{HMH}.
The list decoding of symbol-pair codes was investigated in \cite{LXY}.

There are few works on studying the symbol-pair weight distributions of linear codes.
Symbol-pair distances of a class of repeated-root cyclic codes over $\mathbb{F}_{p^m}$ were determined in \cite{SZW}.
Dinh et al. successfully derived the symbol-pair distances for all constacyclic codes of length $p^s$ and $p^{2s}$ over $\mathbb{F}_{p^m}$ \cite{DNSS1,DWLS}.
The Hamming and symbol-pair distances of all the repeated-root constacyclic codes of length $p^s$ over $\mathcal{R}=\mathbb{F}_{p^m}+u\mathbb{F}_{p^m}$ were presented in \cite{DNSS2}.

In this paper we will focus on the symbol-pair weight distribution of MDS codes and simplex codes.
The contributions of this paper are two-fold.
The first one is to derive the symbol-pair weight distribution of MDS codes.
The second one is to calculate the symbol $b$-weight of codewords in simplex codes for any $b\leq q-1$ and rearranging of coordinates in simplex codes, one may obtain different symbol-pair weight.
In addition, for odd $b$ with $3\leq b\leq q-1$, the symbol $b$-weight of simples code over $\mathbb{F}_{3}$ and for $m\geq 2$, the symbol $p$-weight of codewords in  variation simplex codes over $\mathbb{F}_{p}$ is determined.

The remainder of this paper is organized as follows.
In Section $2$, we introduce some preliminary knowledge and auxiliary results.
The symbol-pair weight distribution of MDS codes are determined in Section $3$.
In Section $4$, we calculate the symbol $b$-weight distribution of simplex codes.
Section $5$ concludes this paper.

\section{Preliminaries}

Let $\mathbb{F}_{q}$ be the finite field with $q$ elements, where $q$ is a prime power.
Let $n$ be a positive integer.
The {\it absolute trace} function ${\rm Tr}\left(x\right)$ from $\mathbb{F}_{p^{m}}$ to $\mathbb{F}_{p}$ is given by
\begin{equation*}
{\rm Tr}\left(x\right)=\sum\limits_{i=0}^{m-1}x^{p^{i}}=x+x^{p}+\cdots+x^{p^{m-1}}.
\end{equation*}
Each element of $\mathbb{F}_{q}$ is called a {\it symbol}.
In the rest of this paper, the subscripts will always be taken by modulo $n$ and $\mathbf{0}$ denotes the all-zero vector.
Let $b$ be a positive integer.
For a vector $\mathbf{x}=\left(x_{0}, x_{1},\cdots, x_{n-1}\right)$, we define the {\it $b$-symbol read vector} of $\mathbf{x}$ as
\begin{equation*}
\pi_{b}\left(\mathbf{x}\right)=\left(\left(x_{0},\cdots,x_{b-1}\right),
\,\left(x_{1},\cdots,x_{b}\right),\cdots,\left(x_{n-1},
\,x_{0},\cdots,x_{b-2}\right)\right).
\end{equation*}
For any two vectors $\mathbf{x},\,\mathbf{y}$ in $\mathbb{F}_{q}^{n}$, we have
\begin{equation*}
\pi_{b}\left(\mathbf{x}+\mathbf{y}\right)=\pi_{b}\left(\mathbf{x}\right)
+\pi_{b}\left(\mathbf{y}\right).
\end{equation*}
Recall that the (\,Hamming\,) {\it distance} $d\left(\mathbf{x},\,\mathbf{y}\right)$ between two vectors $\mathbf{x},\,\mathbf{y}\in \mathbb{F}_{q}^{n}$ is the number of coordinates in which $\mathbf{x}$ and $\mathbf{y}$ differ and the (\,Hamming\,) {\it weight} of a vector $\mathbf{x}$ is the number of nonzero coordinates in $\mathbf{x}$, i.e.,
\begin{equation*}
w_{H}\left(\mathbf{x}\right)=d\left(\mathbf{x},\,\mathbf{0}\right)=
\#\left\{\,i\,|\,x_{i}\neq 0,\,i\in \mathbb{Z}_{n}\right\}
\end{equation*}
where $\mathbb{Z}_{n}$ denotes the ring $\mathbb{Z}/n\mathbb{Z}$.
Accordingly, the {\it symbol $b$-distance} between $\mathbf{x}$ and $\mathbf{y}$ is defined as
\begin{equation*}
d_{b}\left(\mathbf{x},\,\mathbf{y}\right)
=d\left(\pi_{b}(\mathbf{x}),\,\pi_{b}(\mathbf{y})\right)
=\#\left\{\,i\,\big|\,\left(x_{i},\cdots, x_{i+b-1}\right)\neq\left(y_{i},\cdots,y_{i+b-1}\right),\,i\in\mathbb{Z}_{n}\right\}
\end{equation*}
and the {\it symbol $b$-weight} of $\mathbf{x}\in \mathbb{F}_{q}^{n}$ is denoted as
\begin{equation*}
w_{b}(\mathbf{x})=d_{b}\left(\mathbf{x},\,\mathbf{0}\right)
=\#\left\{\,i\,\big|\,\left(x_{i},\cdots,x_{i+b-1}\right)\neq \mathbf{0},\,i\in \mathbb{Z}_{n}\right\}.
\end{equation*}
In particular, when $b=2$, for any two vectors $\mathbf{x},\,\mathbf{y}\in \mathbb{F}_{q}^{n}$, the {\it symbol-pair distance} between $\mathbf{x}$ and $\mathbf{y}$ is
\begin{equation*}
d_{2}\left(\mathbf{x},\,\mathbf{y}\right)
=d\left(\pi_{2}(\mathbf{x}),\,\pi_{2}(\mathbf{y})\right)
=\#\left\{\,i\,\big|\,\left(x_{i},\,x_{i+1}\right)\neq \left(y_{i},\,y_{i+1}\right),\,i\in \mathbb{Z}_{n}\right\}
\end{equation*}
and the {\it symbol-pair weight} of a vector $\mathbf{x}$ is
\begin{equation*}
w_{2}\left(\mathbf{x}\right)=d_{2}\left(\mathbf{x},\mathbf{0}\right)
=\#\left\{\,i\,\big|\,(x_{i},\,x_{i+1})\neq \left(0,\,0\right),\,i\in \mathbb{Z}_{n}\right\}.
\end{equation*}
It is shown in \cite{YBS2} that
\begin{equation*}
d\left(\mathbf{x},\,\mathbf{y}\right)+1\leq d_{2}\left(\mathbf{x},\,\mathbf{y}\right)\leq 2\cdot d\left(\mathbf{x},\,\mathbf{y}\right).
\end{equation*}

A {\it code} $\mathcal{C}$ over $\mathbb{F}_{q}$ of length $n$ is a nonempty subset of $\mathbb{F}_{q}^{n}$.
Elements of $\mathcal{C}$ are called {\it codewords}.
If $\mathcal{C}$ is a subspace of $\mathbb{F}_{q}^{n}$, then $\mathcal{C}$ is called a {\it linear} code.
Denote by $\left[n,\,k,\,d\right]_{q}$ the linear code of length $n$, dimension $k$ and minimum distance $d$ over $\mathbb{F}_{q}$.
The code $\mathcal{C}$ is said to have {\it minimum symbol-pair distance} $d_{2}$ if
\begin{equation*}
d_{2}={\rm min}\left\{d_{2}\left(\mathbf{x},\,\mathbf{y}\right)\,\big|\, \mathbf{x},\,\mathbf{y}\in \mathcal{C}, \mathbf{x}\neq\mathbf{y})\right\}.
\end{equation*}
The minimum symbol-pair distance is an important parameter for a symbol-pair code.
A code $\mathcal{C}$ with minimum symbol-pair distance $d_{2}$ can
correct up to $\lfloor\frac{d_{2}-1}{2}\rfloor$ symbol-pair errors, see \cite{CB1}.
The size of a symbol-pair code satisfies the Singleton bound.

\begin{lem}{\rm(\!{\cite{CJKWY}}\,)}\label{Singleton}
Let $q\geq 2$ and $2\leq d_{2}\leq n$.
If $\mathcal{C}$ is a symbol-pair code with length $n$ and minimum symbol-pair distance $d_{2}$, then $\left|\mathcal{C}\right|\leq q^{n-d_{2}+2}$.
\end{lem}

A symbol-pair code achieving the Singleton bound is called a {\it maximum distance separable} (\,MDS\,) symbol-pair code.
A classical MDS code is also an MDS symbol-pair code.
For any $1<n<q$, we choose $\mathbf{a}=(a_{1},a_{2},\cdots,a_{n})$ where $a_{i}$ are distinct elements in $\mathbb{F}_{q}$.
Then the {\it Reed-Solomon} (\,RS\,) code of length $n$ associated with $\mathbf{a}$ is
\begin{equation*}
\mathbf{RS}_{k}\left(\mathbf{a}\right)=
\left\{\left(f(a_{1}),\cdots,f(a_{n})\right)\,\big|\,f(x)\in \mathbb{F}_{q}[x],\,{\rm deg}\left(f(x)\right)\leq k-1\right\}
\end{equation*}
where $1\leq k\leq n$.
It is known that the code $\mathbf{RS}_{k}(\mathbf{a})$ is an $\left[n,\,k,\,n-k+1\right]_{q}$ MDS code \cite{MS}.

Now denote by $A_{i}$ the number of codewords in $\mathcal{C}$ with weight $i$.
Then the weight distribution of MDS codes is determined in the following lemma.

\begin{lem}{\rm(\!\!{\cite{MS}}\,)}\label{lemwdistribu}
Let $\mathcal{C}$ be an $[n,\,k,\,d\,]_{q}$ MDS code. Then the weight distribution of $\mathcal{C}$ is given by $A_{0}=1$, $A_{i}=0$ for $1\leq i< d$ and
\begin{equation*}
A_{i}=\left(\begin{array}{c}n\\i\end{array}\right)
\sum^{i-d}_{j=0}\left(-1\right)^{j}
\left(\begin{array}{c}i\\j\end{array}\right)
\left(q^{i+1-d-j}-1\right)
\end{equation*}
for $d\leq i\leq n$, where $d=n-k+1$
and $\left(\begin{array}{c}m\\n\end{array}\right)=0$ for $m<n$.
\end{lem}

Inspired by the proof of Lemma \ref{lemwdistribu} in Chap. $7$, pp. $262$ of \cite{MS}, one can calculate the symbol-pair weight distribution of MDS codes employing the shortened codes of MDS codes.

Let $\mathcal{C}$ be a code of length $n$ over $\mathbb{F}_{q}$ and $T\subseteq \left[0,\,n-1\right],\,\overline{T}=\left[0,\,n-1\right]\backslash T.$
Define
\begin{equation*}
\mathcal{C}_{T}=\left\{\left(c_j\right)_{j\in \overline{T}}\,\Big|\,{\rm there\;exists\;some}\;\left(c_{0},\,c_{1},\cdots,c_{n-1}\right)\in \mathcal{C}\;{\rm such\;that}\;c_{i}=0\;{\rm for\; all}\; i\in T\right\}.
\end{equation*}
Then we call $\mathcal{C}_{T}$ the code {\it shortened} on $T$ from $\mathcal{C}$.
For more details, see Section 1.5 in \cite{MS}.
The following lemma characterizes that a code shortened on $T$ from an MDS code is also an MDS code.
It will be applied to calculate the symbol-pair weight distribution of MDS codes in Section $3$.

\begin{lem}{\rm(\!{\cite{MS}}\,)}\label{lemshorten}
Let $\mathcal{C}$ be an $\left[n,\,k,\,n-k+1\right]_{q}$ MDS code and $\mathcal{C}_{T}$ be the code shortened on $T$ from $\mathcal{C}$, where $T$ is a set of $t$ coordinates with $t<k$.
Then $\mathcal{C}_{T}$ is an $\left[n-t,\,k-t,\,n-k+1\right]$ MDS code.
\end{lem}

\section{Symbol-pair weight distribution of MDS codes}

This section will focus on determining the symbol-pair weight distribution of MDS codes.

\begin{thm}\label{thmGRS}
Let $\mathcal{C}$ be an $[n,\,k,\,d]_{q}$ MDS code and $B_{w}$ the number of codewords in $\mathcal{C}$ with symbol-pair weight $w$.
Then
\begin{equation*}
B_{w}=
\left\{
\begin{array}{cl}
1, & w=0,\\[2mm]
0, & 1\leq w\leq d\\
\end{array}
\right.
\end{equation*}
and for $d<w\leq n$,
\begin{equation*}
\begin{split}
B_{w}&=2\sum_{i=1}^{\lfloor M_{1}\rfloor}
\sum^{w-i-d}_{j=0}\left(-1\right)^{j}
\left(\begin{array}{c}n-w+i-1\\i-1\end{array}\right)
\left(\begin{array}{c}w-i-1\\i-1\end{array}\right)
\left(\begin{array}{c}w-i\\j\end{array}\right)
\left(q^{w-i+1-d-j}-1\right)
\\&
+\sum_{i=1}^{\lfloor M_{2}\rfloor}
\sum^{w-i-d}_{j=0}\left(-1\right)^{j}
\left(\begin{array}{c}n-w+i-1\\i-1\end{array}\right)
\left(\begin{array}{c}w-i-1\\i\end{array}\right)
\left(\begin{array}{c}w-i\\j\end{array}\right)
\left(q^{w-i+1-d-j}-1\right)+\varepsilon(w)
\end{split}
\end{equation*}
where
\begin{equation*}
d=n-k+1,\,M_{1}={\rm min}\left\{\frac{w}{2},\,w-d\right\},
\,M_{2}={\rm min}\left\{\frac{w-1}{2},\,w-d\right\}
\end{equation*}
and
\begin{equation*}
\varepsilon(w)=
\left\{
\begin{array}{ll}
\sum_{i=1}^{\lfloor M_{1}\rfloor}
\sum^{w-i-d}_{j=0}\left(-1\right)^{j}
\left(\begin{array}{c}n-w+i-1\\i\end{array}\right)
\left(\begin{array}{c}w-i-1\\i-1\end{array}\right)\cdot&\\
\qquad\qquad\quad\left(\begin{array}{c}w-i\\j\end{array}\right)
\left(q^{w-i+1-d-j}-1\right), & {\rm if}\; d<w\leq n-1;\\[6mm]
\sum^{n-d}_{j=0}\left(-1\right)^{j}
\left(\begin{array}{c}n\\j\end{array}\right)
\left(q^{n+1-d-j}-1\right), & {\rm if}\; w=n.\\
\end{array}
\right.
\end{equation*}
\end{thm}

\begin{proof}
It is easily verified that $B_{0}=1$ and $B_{w}=0$ for any $1\leq w\leq d$.
In the following, we just need to consider the case $d<w\leq n$.
For convenience, we define the following notations.
A codeword $\mathbf{c}=\left(c_{1},\cdots,c_{n}\right)$ is called {\it Type\,-$0$}\,(\,{\it Type}\,-$1$, resp.) if $c_{1}=0$\,(\,$c_{1}\neq 0$, resp.).
Decompose
\begin{equation*}
\left[1,\,n\right]=N_{1}\bigcup N_{2}\bigcup\cdots\bigcup N_{l}
\end{equation*}
where
\begin{equation*}
N_{i}=\left[n_{1}+\cdots+n_{i-1}+1,n_{1}+\cdots+n_{i}\right]
\end{equation*}
with $n_{0}=0$, $\sum^{l}_{i=1}{n_i}=n$ and $|N_{i}|=n_{i}\geq1$ for each $1\leq i\leq l$.
Then we call a codeword $\mathbf{c}$ having {\it Shape} $\left(N_{1},\,N_{2},\cdots,N_{l}\right)$ if its entries satisfy:
\begin{itemize}
\item
for any $j\in N_{i}$, $c_{j}$ is always zero\,(\,nonzero, resp.);
\item
for $1\leq i\leq l$, if $c_{j}$ is zero (nonzero, resp.) for any $j\in N_{i}$, then $c_{j}$ is nonzero (zero, resp.) for any $j\in N_{i+1}$.
\end{itemize}
Now the problem on deriving $B_{w}$, i.e., the number of codewords in $\mathcal{C}$ with symbol-pair weight $w$, can be transformed to enumerate the codewords in $\mathcal{C}$ with  Shape $\left(N_{1},\,N_{2},\cdots,N_{l}\right)$ and Type\,-$0$ (or Type\,-$1$).
\medskip
\\
$\mathbf{Case\,I\,}$ (\,$l$\, even) :
We first consider the subcase of Type\,-$0$.
In this subcase, for a codeword $\mathbf{c}$ with symbol-pair weight $w$ and Shape $\left(N_{1},\,N_{2},\cdots,N_{l}\right)$, we have
\begin{equation*}
w=n_{2}+n_{4}+\cdots+n_{l}+\frac{l}{2}.
\end{equation*}
Let $P_{1}$ be the number of Shape $\left(N_{1},\,N_{2},\cdots,N_{l}\right)$ such that the corresponding vector $\mathbf{c}\in \mathbb{F}_q^n$ of Type\,-$0$ has symbol-pair weight $w$.
It is straightforward that
\begin{equation*}
\left\{
\begin{array}{l}
n_{1}+n_{2}+\cdots+n_{l}=n,\\[2mm]
n_{2}+n_{4}+\cdots+n_{l}+\frac{l}{2}=w,
\end{array}
\right.
\end{equation*}
which is equivalent to
\begin{equation*}
\left\{
\begin{array}{l}
n_{1}+n_{3}+\cdots+n_{l-1}=n-w+\frac{l}{2},\\[2mm]
n_{2}+n_{4}+\cdots+n_{l}=w-\frac{l}{2}.
\end{array}
\right.
\end{equation*}
Note that
\begin{equation*}
\left\{
\begin{array}{l}
w-\frac{l}{2}\geq \frac{l}{2},\\[2mm]
w-\frac{l}{2}\geq d,\\
\end{array}
\right.
\end{equation*}
which implies that
\begin{equation*}
\frac{l}{2}\leq {\rm min}\left\{\frac{w}{2},\,w-d\right\}=M_{1}.
\end{equation*}
Counting arguments lead to
\begin{equation*}
P_{1}=\sum_{\frac{l}{2}=1}^{\lfloor{M_{1}}\rfloor}
\left(\begin{array}{c}n-w+\frac{l}{2}-1\\\frac{l}{2}-1\end{array}\right)
\left(\begin{array}{c}w-\frac{l}{2}-1\\\frac{l}{2}-1\end{array}\right).
\end{equation*}
In the sequel, we concentrate on determining the number of codewords $\mathbf{c}$ in $\mathcal{C}$ with symbol-pair weight $w$ and Type\,-$0$ corresponding to each Shape $(N_{1},\,N_{2},\cdots,N_{l})$.
Let
\begin{equation*}
T_{1}=N_{1}\bigcup N_{3}\bigcup\cdots\bigcup N_{l}
\end{equation*}
and $\mathcal{C}_{T_{1}}$ be the subcode of $\mathcal{C}$ shortened on $T_{1}$.
It follows from Lemma \ref{lemshorten} that $\mathcal{C}_{T_{1}}$ is a
\begin{equation*}
\left[w-\frac{l}{2},\,w-\frac{l}{2}+k-n,\,n-k+1\right]
\end{equation*}
MDS code since $\left|T_{1}\right|=n-w+\frac{l}{2}$.
According to Lemma \ref{lemwdistribu}, the number of codewords in $\mathcal{C}^{'}$ with weight $w-\frac{l}{2}$ is
\begin{equation*}
A_{w-\frac{l}{2}}^{'}=\sum^{w-\frac{l}{2}-d}_{j=0}\left(-1\right)^{j}
\left(\begin{array}{c}w-\frac{l}{2}\\j\end{array}\right)
\left(q^{w-\frac{l}{2}+1-d-j}-1\right).
\end{equation*}
It follows that the number of codewords in $\mathcal{C}$ with symbol-pair weight $w$ and Type\,-$0$ corresponding to each Shape $\left(N_{1},\,N_{2},\cdots,N_{l}\right)$ is $A_{w-\frac{l}{2}}^{'}$.

Note that $A_{w-\frac{l}{2}}^{'}$ only relies on $w$ and $l$, not on specific choices of $N_{1},\,N_{2},\cdots,N_{l}$.
Thus for even $l$, the number of codewords in $\mathcal{C}$ with symbol-pair weight $w$ and Type\,-$0$ is
\begin{equation*}
\begin{split}
B_{w}^{\left(1\right)}&=P_{1}\cdot A_{w-\frac{l}{2}}^{'}
\\&
=\sum_{\frac{l}{2}=1}^{\lfloor M_{1}\rfloor}
\left(\begin{array}{c}n-w+\frac{l}{2}-1\\\frac{l}{2}-1\end{array}\right)
\left(\begin{array}{c}w-\frac{l}{2}-1\\\frac{l}{2}-1\end{array}\right)
\sum^{w-\frac{l}{2}-d}_{j=0}\left(-1\right)^{j}
\left(\begin{array}{c}w-\frac{l}{2}\\j\end{array}\right)
\left(q^{w-\frac{l}{2}+1-d-j}-1\right)
\\&
=\sum_{i=1}^{\lfloor M_{1}\rfloor}
\sum^{w-i-d}_{j=0}\left(-1\right)^{j}
\left(\begin{array}{c}n-w+i-1\\i-1\end{array}\right)
\left(\begin{array}{c}w-i-1\\i-1\end{array}\right)
\left(\begin{array}{c}w-i\\j\end{array}\right)
\left(q^{w-i+1-d-j}-1\right).
\end{split}
\end{equation*}
Similarly, it can be verified that for even $l$, the number of codewords in $\mathcal{C}$ with symbol-pair weight $w$ and Type\,-$1$ is also $B_{w}^{(1)}$.
Therefore, for even $l$, the number of codewords in $\mathcal{C}$ with symbol-pair weight $w$ is equal to $2B_{w}^{(1)}$.
\medskip
\\
$\mathbf{Case\,II\,}$ (\,$l$\, odd) :
A Type\,-$0$ codeword $\mathbf{c}$ having symbol-pair weight $w$ and Shape $\left(N_{1},\,N_{2},\cdots,N_{l}\right)$ satisfies
\begin{equation*}
\left\{
\begin{array}{l}
n_{1}+n_{2}+\cdots+n_{l}=n,\\[2mm]
n_{2}+n_{4}+\cdots+n_{l-1}+\frac{l-1}{2}=w,
\end{array}
\right.
\end{equation*}
which yields
\begin{equation}\label{eqn_i2}
\left\{
\begin{array}{l}
n_{1}+n_{3}+\cdots+n_{l}=n-w+\frac{l-1}{2},\\[2mm]
n_{2}+n_{4}+\cdots+n_{l-1}=w-\frac{l-1}{2}.
\end{array}
\right.
\end{equation}
It is worth noting that
\begin{equation*}
\left\{
\begin{array}{l}
n-w+\frac{l-1}{2}\geq \frac{l-1}{2}+1,\\[2mm]
w-\frac{l-1}{2}\geq \frac{l-1}{2},\\[2mm]
w-\frac{l-1}{2}\geq d,\\
\end{array}
\right.
\end{equation*}
which indicates that $w\leq n-1$ and
\begin{equation*}
\frac{l-1}{2}\leq M_{1}.
\end{equation*}
Hence the number of Shape $\left(N_{1},\,N_{2},\cdots,N_{l}\right)$ satisfying (\ref{eqn_i2}) is
\begin{equation*}
P_{2}=\delta(w)\sum_{\frac{l-1}{2}=1}^{\lfloor M_{1}\rfloor}
\left(\begin{array}{c}n-w+\frac{l-1}{2}-1\\\frac{l-1}{2}\end{array}\right)
\left(\begin{array}{c}w-\frac{l-1}{2}-1\\\frac{l-1}{2}-1\end{array}\right)
\end{equation*}
where
\begin{equation*}
\delta(w)=
\left\{
\begin{array}{cl}
1, & w\leq n-1,\\[2mm]
0, & w=n.\\
\end{array}
\right.
\end{equation*}
Denote
\begin{equation*}
T_{2}=N_{1}\bigcup N_{3}\bigcup\cdots\bigcup N_{l-2}\bigcup N_{l}.
\end{equation*}
Then
\begin{equation*}
\left|T_{2}\right|=n-w+\frac{l-1}{2}.
\end{equation*}
Let $\mathcal{C}_{T_{2}}$ be the subcode of $\mathcal{C}$ shortened on $T_{2}$.
By Lemma \ref{lemshorten}, one can obtain that $\mathcal{C}_{T_{2}}$ is a
\begin{equation*}
\left[w-\frac{l-1}{2},\,w-\frac{l-1}{2}+k-n,\,n-k+1\right]
\end{equation*}
MDS code.
Due to Lemma \ref{lemwdistribu}, the number of codewords in $\mathcal{C}_{T_{2}}$ with weight $w-\frac{l-1}{2}$ is
\begin{equation*}
A_{w-\frac{l-1}{2}}^{''}=\sum^{w-\frac{l-1}{2}-d}_{j=0}\left(-1\right)^{j}
\left(\begin{array}{c}w-\frac{l-1}{2}\\j\end{array}\right)
\left(q^{w-\frac{l-1}{2}+1-d-j}-1\right).
\end{equation*}
Therefore, for odd $l$, the number of codewords in $\mathcal{C}$ with symbol-pair weight $w$ and Type\,-$0$ is
\begin{equation*}
\begin{split}
B_{w}^{(2)}&=P_{2}\cdot A_{w-\frac{l-1}{2}}^{''}
\\&
=\delta(w)\sum_{i=1}^{\lfloor M_{1}\rfloor}
\sum^{w-i-d}_{j=0}\left(-1\right)^{j}
\left(\begin{array}{c}n-w+i-1\\i\end{array}\right)
\left(\begin{array}{c}w-i-1\\i-1\end{array}\right)
\left(\begin{array}{c}w-i\\j\end{array}\right)
\left(q^{w-i+1-d-j}-1\right).
\end{split}
\end{equation*}
It can be similarly shown that, for odd $l$, the number of codewords in $\mathcal{C}$ with symbol-pair weight $w$ and Type\,-$1$ is
\begin{equation*}
B_{w}^{(3)}=\sum_{i=1}^{\lfloor M_{2}\rfloor}
\sum^{w-i-d}_{j=0}\left(-1\right)^{j}
\left(\begin{array}{c}n-w+i-1\\i-1\end{array}\right)
\left(\begin{array}{c}w-i-1\\i\end{array}\right)
\left(\begin{array}{c}w-i\\j\end{array}\right)
\left(q^{w-i+1-d-j}-1\right)
\end{equation*}
where
\begin{equation*}
M_{2}={\rm min}\left\{\frac{w-1}{2},\,w-d\right\}.
\end{equation*}
\medskip

As a result, the total number of codewords in $\mathcal{C}$ with symbol-pair weight $w$ can be derived:
\begin{equation*}
B_{w}=
\left\{
\begin{array}{ll}
2B_{w}^{(1)}+B_{w}^{(2)}+B_{w}^{(3)}, & d<w<n,\\[2mm]
2B_{w}^{(1)}+B_{w}^{(3)}+A_{n}, & w=n,\\
\end{array}
\right.
\end{equation*}
where
\begin{equation*}
A_{n}=\sum^{n-d}_{j=0}\left(-1\right)^{j}
\left(\begin{array}{c}n\\j\end{array}\right)
\left(q^{n+1-d-j}-1\right).
\end{equation*}
This completes the desired conclusion.
\end{proof}

In the sequel, some examples of Theorem \ref{thmGRS} on MDS codes are presented.

\begin{example}
Let $\mathcal{C}$ be a $\left[4,\,3,\,2\right]_{q}$ MDS code.
Then by Theorem \ref{thmGRS},
\begin{equation*}
B_{0}=1,\,B_{1}=B_{2}=0,\,B_{3}
=2\left(q-1\right)+\left(q-1\right)+\left(q-1\right)=4q-4
\end{equation*}
and
\begin{equation*}
\begin{split}
B_{4}=&\sum_{j=0}^{1}(-1)^{j}
\left(\begin{array}{c}2\\1\end{array}\right)
\left(\begin{array}{c}3\\j\end{array}\right)
\left(q^{2-j}-1\right)
+\sum_{j=0}^{2}(-1)^{j}
\left(\begin{array}{c}4\\j\end{array}\right)
\left(q^{3-j}-1\right)
\\&
+2\sum_{i=1}^{2}\sum_{j=0}^{2-i}\left(-1\right)^{j}
\left(\begin{array}{c}3-i\\i-1\end{array}\right)
\left(\begin{array}{c}4-i\\j\end{array}\right)
\left(q^{3-i-j}-1\right)
\\
=&\,q^{3}-4q+3.
\end{split}
\end{equation*}
For the case of $q=8$, let $\mathbf{a}=(1,\theta,\,\theta^2,\,\theta^3)$
where $\theta$ is a primitive element of $\mathbb{F}_{q}^{*}$ with $\theta^3+\theta+1=0$.
By computation software MAGMA, the code $\mathbf{RS}_{3}\left(\mathbf{a}\right)$ has symbol-pair weight distribution:
\begin{equation*}
B_{0}=1,\,B_{1}=B_{2}=0,\,B_{3}=28,\,B_{4}=483,
\end{equation*}
which coincides with the result in Theorem \ref{thmGRS}.
\end{example}

\begin{example}
Let $\mathcal{C}$ be a $[5,\,4,\,2]_{q}$ MDS code.
Then Theorem \ref{thmGRS} yields that
\begin{equation*}
B_{0}=1,\,B_{1}=B_{2}=0,\,B_{3}=2(q-1)+2(q-1)+(q-1)=5q-5,
\end{equation*}
\begin{equation*}
\begin{split}
B_{4}=&\sum_{j=0}^{1}(-1)^{j}
\left(\begin{array}{c}2\\1\end{array}\right)
\left(\begin{array}{c}3\\j\end{array}\right)
\left(q^{2-j}-1\right)
+\sum_{i=1}^{2}\sum_{j=0}^{2-i}(-1)^{j}
\left(\begin{array}{c}3-i\\i-1\end{array}\right)
\left(\begin{array}{c}4-i\\j\end{array}\right)
\left(q^{3-i-j}-1\right)
\\&
+2\sum_{i=1}^{2}\sum_{j=0}^{2-i}(-1)^{j}
\left(\begin{array}{c}i\\i-1\end{array}\right)
\left(\begin{array}{c}3-i\\i-1\end{array}\right)
\left(\begin{array}{c}4-i\\j\end{array}\right)
\left(q^{3-i-j}-1\right)
\\
=&\,5q^{2}-10q+5
\end{split}
\end{equation*}
and
\begin{equation*}
\begin{split}
B_{5}=&\sum_{j=0}^{3}(-1)^{j}
\left(\begin{array}{c}5\\j\end{array}\right)
\left(q^{4-j}-1\right)
+\sum_{i=1}^{2}\sum_{j=0}^{3-i}(-1)^{j}
\left(\begin{array}{c}4-i\\i\end{array}\right)
\left(\begin{array}{c}5-i\\j\end{array}\right)
\left(q^{4-i-j}-1\right)
\\&
+2\sum_{i=1}^{2}\sum_{j=0}^{3-i}(-1)^{j}
\left(\begin{array}{c}4-i\\i-1\end{array}\right)
\left(\begin{array}{c}5-i\\j\end{array}\right)
\left(q^{4-i-j}-1\right)
\\
=&\,q^{4}-5q^{2}+5q-1.
\end{split}
\end{equation*}
For the case of $q=27$, let $\mathbf{a}=(1,\,\theta,\,\theta^2,\,\theta^3,\,\theta^4)$
where $\theta$ is a primitive element of $\mathbb{F}_{q}^{*}$ with $\theta^3-\theta+1=0$.
It can be verified by MAGMA that the code $\mathbf{RS}_{4}\left(\mathbf{a}\right)$ has symbol-pair weight distribution:
\begin{equation*}
B_{0}=1,\,B_{1}=B_{2}=0,\,B_{3}=130,\,B_{4}=3380,\,B_{5}=527930,
\end{equation*}
which coincides with the result in Theorem \ref{thmGRS}.
\end{example}

\section{Symbol $b$-weight distribution of simplex codes}

In this section, we shall analyze symbol $b$-weight distribution of simplex codes for any $b$.
We show that all the nonzero codewords in a simplex code have the same symbol $b$-weight.
Surprisingly,  rearranging coordinates of a simplex code may lead to different symbol $b$-weight.
From now on, we always set $h=\frac{q-1}{p-1}$ and $g$ a primitive element of $\mathbb{F}_{q}^{*}$.
For brevity, denote by
\begin{equation*}
\mathcal{C}=\left\{\mathbf{c}_{\alpha}=\left({\rm Tr}\left(\alpha\right),{\rm Tr}\left(g\alpha\right),\cdots,{\rm Tr}\left(g^{q-2}\alpha\right)\right)
\,\big|\,\alpha\in\mathbb{F}_{q}\right\}
\end{equation*}
the cyclic simplex code of length $q-1$ and
\begin{equation}\label{eqvariationsc}
\mathcal{C^{'}}=\left\{\mathbf{c}_{\alpha}^{'}=
\left(\overline{\mathbf{c}}_{\alpha}^{\left(0\right)}\,\big|\,
\overline{\mathbf{c}}_{\alpha}^{\left(1\right)}\,\big|\,\cdots\,\big|\,
\overline{\mathbf{c}}_{\alpha}^{\left(h-1\right)}\right)
\,\Big|\,\alpha\in\mathbb{F}_{q}\right\}
\end{equation}
the variation simplex code, where
\begin{equation}\label{calphai}
\overline{\mathbf{c}}_{\alpha}^{\left(i\right)}=\left({\rm Tr}\left(g^{i}\alpha\right),{\rm Tr}\left(g^{i+h}\alpha\right),\cdots,{\rm Tr}\left(g^{i+(p-2)h}\alpha\right)\right).
\end{equation}

Firstly, the symbol $b$-weight distribution of cyclic simplex code is given as follows.

\begin{thm}\label{thmb}
Let $q=p^m$.  Then the cyclic simplex code $\mathcal{C}$ has only one nonzero symbol $b$-weight
\begin{equation*}
\left\{
\begin{array}{ll}
q-p^{m-b}, & b<m,\\[2mm]
q-1, & m\leq b\leq q-1.\\
\end{array}
\right.
\end{equation*}
\end{thm}

\begin{proof}
Let
\begin{equation*}
L_{i}=\left\{x\in\mathbb{F}_{q}\,\big|\,{\rm Tr}\left(g^{i-1}x\right)=0\right\}
\end{equation*}
where $1\leq i\leq b$.
Note that $L_{i}$ is an $\mathbb{F}_{p}$-linear subspace in $\mathbb{F}_q$ with codimension $1$.
It follows that the symbol $b$-weight of any nonzero codeword $\mathbf{c}_{\alpha}=\left({\rm Tr}\left(\alpha\right),{\rm Tr}\left(g\alpha\right),\cdots,{\rm Tr}\left(g^{q-2}\alpha\right)\right)$ in $\mathcal{C}$ is
\begin{equation*}
\begin{split}
w_{b}\left(\mathbf{c}_{\alpha}\right)&=q-1-\#\left\{0\leq i\leq q-2\,\big|\,{\rm Tr}(g^{i}\alpha)={\rm Tr}(g^{i+1}\alpha)=\cdots={\rm Tr}(g^{i+b-1}\alpha)=0\right\}
\\[1mm]&
=q-1-\#\left\{x\in\mathbb{F}_{q}^{*}\,\big|\,x\in L_{1}\cap L_{2}\cap \cdots \cap L_{b} \right\}
\\[2mm]&
=\left\{
\begin{array}{ll}
q-1-\left(p^{m-b}-1\right), & b<m,\\[2mm]
q-1-\left(1-1\right), & m\leq b\leq q-1.\\
\end{array}
\right.
\\[2mm]&
=\left\{
\begin{array}{ll}
q-p^{m-b}, & b<m,\\[2mm]
q-1, & m\leq b\leq q-1,\\
\end{array}
\right.
\end{split}
\end{equation*}
Here the third equality follows from that for $b<m$, the $\mathbb{F}_{p}$-linear system of equations
\begin{equation}\label{eqlinear}
\left\{
\begin{array}{l}
{\rm Tr}(g^{i}\alpha)=0,\\
\vdots\\
{\rm Tr}(g^{i+b-1}\alpha)=0
\end{array}
\right.
\end{equation}
has $p^{m-b}$ solutions since $g^{i},\,\cdots,g^{i+b-1}$ are linear independent over $\mathbb{F}_{p}$.
If $m\leq b\leq q-1$, the equation (\ref{eqlinear}) has only one solution $\alpha=0$.
The proof is completed.
\end{proof}

\begin{example}
Let $p=2$, $m=3$ and $g$ be a primitive element of $\mathbb{F}_{q}^{*}$ with $g^3+g^2+1=0$.
For the simplex code
\begin{equation*}
\mathcal{C}=\left\{\mathbf{c}_{\alpha}=\left({\rm Tr}\left(\alpha\right),{\rm Tr}\left(g\alpha\right),{\rm Tr}\left(g^2\alpha\right),{\rm Tr}\left(g^3\alpha\right),{\rm Tr}\left(g^4\alpha\right),{\rm Tr}\left(g^5\alpha \right),{\rm Tr}\left(g^6\alpha\right)\right)
\,\big|\,\alpha\in\mathbb{F}_{q}\right\},
\end{equation*}
we choose $\mathbf{c}_{g}=\left(1,1,0,1,0,0,1\right)$.
It can be checked by Magma that
\begin{equation*}
w_{1}\left(\mathbf{c}_{g}\right)=4,\,w_{2}\left(\mathbf{c}_{g}\right)=6,\,
w_{b}\left(\mathbf{c}_{g}\right)=7\;for\;3\leq b\leq 7,
\end{equation*}
which coincides with Theorem \ref{thmb}.
\end{example}

\begin{remark}\label{thms}
Any nonzero codeword of the cyclic simplex code has symbol-pair weight $q-p^{m-2}$, which follows directly from Theorem \ref{thmb} by taking $b=2$.
\end{remark}

The following corollary presents symbol $b$-weight distribution of a standard simplex code, which is useful in the sequel.

\begin{cor}\label{cor1}
Let $q=p^m$ and
\begin{equation*}
\mathcal{C}^{''}=\left\{\mathbf{c}_{\alpha}^{''}=\left({\rm Tr}\left(\alpha\right),\,{\rm Tr}\left(g\alpha\right),\cdots,{\rm Tr}\left(g^{\frac{q-1}{p-1}-1}\alpha\right)\right)
\,\Big|\,\alpha\in\mathbb{F}_{q}\right\}
\end{equation*}
be a standard simplex code of length $\frac{q-1}{p-1}$.
Then all nonzero codewords $\mathbf{c}_{\alpha}^{''}$ in $\mathcal{C}^{''}$ have symbol $b$-weight
\begin{equation*}
\left\{
\begin{array}{ll}
\frac{q-p^{m-b}}{p-1}, & b<m,\\[2mm]
\frac{q-1}{p-1}, & m\leq b\leq q-1.\\
\end{array}
\right.
\end{equation*}
\end{cor}

\begin{proof}
Let $u=g^{\frac{q-1}{p-1}}$.
Note that for any nonzero $\alpha\in \mathbb{F}_{q}$,
\begin{equation*}
\begin{split}
\mathbf{c}_{\alpha}&=\left(\mathbf{c}_{\alpha}^{''}\,\big|\,
\mathbf{c}_{u\alpha}^{''}\,\big|\,\cdots\,
\big|\,\mathbf{c}_{u^{p-2}\alpha}^{''}\right)
\\[2mm] &
=\left(\mathbf{c}_{\alpha}^{''}\,\big|\,u\mathbf{c}_{\alpha}^{''}
\,\big|\,\cdots\,\big|\,u^{p-2}\mathbf{c}_{\alpha}^{''}\right).
\end{split}
\end{equation*}
It follows that
\begin{equation*}
w_{b}\left(\mathbf{c}_{\alpha}\right)=\left(p-1\right)
w_{b}\left(\mathbf{c}_{\alpha}^{''}\right).
\end{equation*}
Hence
\begin{equation*}
w_{b}\left(\mathbf{c}_{\alpha}^{''}\right)=
\frac{1}{p-1}w_{b}\left(\mathbf{c}_{\alpha}\right).
\end{equation*}
The desired conclusion follows from Theorem \ref{thmb}.
\end{proof}

In the sequel, we will calculate the symbol-pair weight distribution of variation simplex codes.

\begin{thm}\label{thmvsp}
Let $q=p^m$.
Then the symbol-pair weight of any nonzero codeword of $\mathcal{C^{'}}$ in (\ref{eqvariationsc}) is $q-p^{m-1}+p^{m-2}$.
\end{thm}

\begin{proof}
Now for
\begin{equation*}
\mathbf{c}_{\alpha}^{''}=\left({\rm Tr}(\alpha),\,{\rm Tr}(g\alpha),\cdots,{\rm Tr}\left(g^{h-1}\alpha\right)\right),
\end{equation*}
we claim that
\begin{equation*}
w_{H}\left(\mathbf{c}_{\alpha}^{''}\right)=p^{m-1}.
\end{equation*}
Indeed, the number of $x\in \mathbb{F}_{q}^{*}$ with ${\rm Tr}(x)=0$ is  $p^{m-1}-1$.
It can be verified that the number of zero entries in $\mathbf{c}_{\alpha}^{''}$ is $\frac{p^{m-1}-1}{p-1}$.
Then one can obtain that
\begin{equation*}
w_{H}\left(\mathbf{c}_{\alpha}^{''}\right)
=\frac{q-1}{p-1}-\frac{p^{m-1}-1}{p-1}=p^{m-1}.
\end{equation*}
It follows from Corollary \ref{cor1} that
\begin{equation*}
w_{2}\left(\mathbf{c}_{\alpha}^{''}\right)=p^{m-1}+p^{m-2}.
\end{equation*}
On the other hand, it is worth noting that
\begin{equation*}
{\rm Tr}\left(g^{i+hj}\alpha\right)=g^{hj}{\rm Tr}\left(g^{i}\alpha\right)
\end{equation*}
for any $1\leq j\leq p-2$.
Hence for any $0\leq i\leq h-1$, all the entries of  $\overline{\mathbf{c}}_{\alpha}^{\left(i\right)}$ in (\ref{calphai}) are zero (or nonzero).
Recall the definition of Shape in the proof of Theorem \ref{thmGRS}.
Let
\begin{equation*}
\mathbf{c}_{\alpha}^{''}=\left({\rm Tr}(\alpha),\,{\rm Tr}(g\alpha),\cdots,{\rm Tr}\left(g^{h-1}\alpha\right)\right)
\end{equation*}
 be a codeword of Shape $(N_1,N_2,\cdots, N_l)$ with $|N_i|=n_i$ and $\sum\limits_{i=1}^l=h$.
 Note that
 \begin{itemize}
   \item for any $N_i$ such that ${\rm Tr}\left(g^j\alpha\right)$ is always nonzero for any $j\in N_i$, concatenating $\left(p-2\right)$ nonzero elements  at the end of each ${\rm Tr}\left(g^j\alpha\right)$ will result in the augment of its original symbol-pair weight by $(p-2)$;
   \item for any $N_i$ such that ${\rm Tr}\left(g^j\alpha\right)$ is always zero for any $j\in N_i$, concatenating $\left(p-2\right)$ zeros at the end of each ${\rm Tr}\left(g^j\alpha\right)$ will keep its original symbol-pair weight.
 \end{itemize}
Therefore, for any nonzero codeword
\begin{equation*}
\mathbf{c}_{\alpha}^{'}=\left(\overline{\mathbf{c}}_{\alpha}^{\left(0\right)}
\,\big|\,\overline{\mathbf{c}}_{\alpha}^{\left(1\right)}\,\big|\,\cdots\,\big|\,
\overline{\mathbf{c}}_{\alpha}^{\left(h-1\right)}\right),
\end{equation*}
one yields
\begin{equation*}
\begin{split}
w_{2}\left(\mathbf{c}_{\alpha}^{'}\right)&=
w_{2}\left(\mathbf{c}_{\alpha}^{''}\right)+
(p-2)w_{H}\left(\mathbf{c}_{\alpha}^{''}\right)
\\[1mm]&
=p^{m-1}+p^{m-2}+\left(p-2\right)p^{m-1}
\\[1mm]&
=q-p^{m-1}+p^{m-2}.
\end{split}
\end{equation*}
Thus we complete the proof.
\end{proof}

The following example illustrates that rearranging the coordinates of simplex codes may induce different symbol-pair weight from Theorems \ref{thmb} and \ref{thmvsp}.

\begin{example}\label{eg1}
Let $p=m=3$ and $g$ be a primitive element of $\mathbb{F}_{q}^{*}$ with $g^3-g+1=0$.
Then for the codewords
\begin{equation*}
\mathbf{c}_{g}=\left(0,-1,0,-1,1,-1,-1,1,0,-1,-1,
-1,0,0,1,0,1,-1,1,1,-1,0,1,1,1,0\right)\in \mathcal{C}
\end{equation*}
and
\begin{equation*}
\mathbf{c}_{g}^{'}=\left(0,0,-1,1,0,0,-1,1,1,-1,-1,1,-1,
1,1,-1,0,0,-1,1,-1,1,-1,1,0,0\right)\in \mathcal{C}^{'},
\end{equation*}
the symbol-pair weight of $\mathbf{c}_{g}$ and $\mathbf{c}_{g}^{'}$ are $24$ (\,see Theorem \ref{thmb}) and $21$ (\,see Theorem \ref{thmvsp}), respectively.
Therefore, rearrangement of coordinates in a simplex code may change the symbol-pair weight.
\end{example}

In what follows, for odd $b$ with $3\leq b\leq m$, the symbol $b$-weight of nonzero codewords in variation simplex codes over $\mathbb{F}_{3}$ is investigated.

\begin{prop}\label{prop33}
Let $q=3^m$.
{Then the symbol $(2s+1)$-weight of any nonzero codewords $\mathbf{c}_{\alpha}^{'}$ in the variation simplex code $\mathcal{C}^{'}$ is}
\begin{equation*}
\left\{
\begin{array}{ll}
q-p^{m-s-1}, & 1\leq s<m,\\[2mm]
q-1, & m\leq s\leq \lfloor \frac{q-2}{2}\rfloor.\\
\end{array}
\right.
\end{equation*}
\end{prop}

\begin{proof}
For $p=3$, it is obvious that both entries of $\overline{\mathbf{c}}_{\alpha}^{\left(i\right)}$ are zero (or nonzero).
{Hence, one obtains
\begin{equation*}
\begin{split}
w_{2s+1}\left(\mathbf{c}_{\alpha}^{'}\right)
&=q-1-2\cdot\#\left\{0\leq i\leq h-1\,\bigg|\,{\rm Tr}\left(g^{i}\alpha\right)={\rm Tr}\left(g^{i+1}\alpha\right)=\cdots={\rm Tr}\left(g^{i+s}\alpha\right)=0\right\}
\\[1mm]&
=q-1-\left(p-1\right)\left(h-w_{s+1}\left(\mathbf{c}_{\alpha}^{''}\right)\right)
\\[1mm]&
=\left(p-1\right)w_{s+1}\left(\mathbf{c}_{\alpha}^{''}\right)
\\[1mm]&
=w_{s+1}\left(\mathbf{c}_{\alpha}\right)
\\[1mm]&
=\left\{
\begin{array}{ll}
q-p^{m-s-1}, & 1\leq s<m,\\[2mm]
q-1, & m\leq s\leq  \frac{q-3}{2}.\\
\end{array}
\right.
\end{split}
\end{equation*}
Here the last equality follows from Theorem \ref{thmb}.}
This completes the proof.
\end{proof}

Now we focus on the symbol $p$-weight of nonzero codewords in the variation simplex code over
$\mathbb{F}_{p}$ with  $p$ odd prime in a similar way of Proposition \ref{prop33}.

\begin{prop}\label{propbeqp}
Let $p$ be odd prime and $q=p^m$ with $m>1$.
The symbol $p$-weight of any nonzero codewords in the variation simplex code $\mathcal{C^{'}}$ in (\ref{eqvariationsc}) is $q-p^{m-2}$.
\end{prop}

\begin{proof}
When $m\geq 2$, by Corollary \ref{cor1}, one obtains
\begin{equation*}
\begin{split}
w_{p}\left(\mathbf{c}_{\alpha}^{'}\right)
&=q-1-\left(p-1\right)\cdot\#\left\{0\leq i\leq h-1\,\bigg|\,{\rm Tr}\left(g^{i}\alpha\right)={\rm Tr}\left(g^{i+1}\alpha\right)=0\right\}
\\[1mm]&
=q-1-\left(p-1\right)\left(h-w_{2}\left(\mathbf{c}_{\alpha}^{''}\right)\right)
\\[1mm]&
=q-1-\left(p-1\right)\left(\frac{q-1}{p-1}-\frac{q-p^{m-2}}{p-1}\right)
\\[1mm]&
=q-p^{m-2}.
\end{split}
\end{equation*}
This completes the proof.
\end{proof}

\begin{remark}
Note that by Theorem \ref{thmb}, Propositions \ref{prop33} and \ref{propbeqp}, for odd $b$ with $3\leq b\leq q-1$, the symbol $b$-weight of nonzero codewords in $\mathcal{C}$ and $\mathcal{C^{'}}$ over $\mathbb{F}_{3^{m}}$ are different.
And for $m\geq p>2$, the symbol $p$-weight of nonzero codewords in $\mathcal{C}$ and $\mathcal{C^{'}}$ over $\mathbb{F}_{p}$ are different as $q-p^{m-p}\neq q-p^{m-2}.$

Furthermore, for $p=3$ and $m\geq 3$, it can be checked that the result in Proposition \ref{prop33} coincides with Proposition \ref{propbeqp}.
However, for even $b$ with $4\leq b\leq q-1$, $p=3$ and for $b\ne p$ with $3\leq b\leq q-1$, $p>3$, it seems not applicable to determine the symbol $b$-weight of codewords in $\mathcal{C^{'}}$ using a similar method in Propositions \ref{prop33} and \ref{propbeqp}.
\end{remark}

\section{Conclusions}

In this paper, the symbol-pair weight distributions of MDS codes and simplex codes are characterized.
By enumeration of different Shapes and utilizing shortened code of MDS codes, one derives the symbol-pair weight distribution of MDS codes.
For any $b$, the symbol $b$-weight distribution of simplex codes is presented and rearranging entries of simplex codes may lead to different symbol-pair weight.
In addition, for odd $b$ with $3\leq b\leq q-1$, the symbol $b$-weight of simplex code over $\mathbb{F}_{3}$ and for any $m>1$, the symbol $p$-weight of codewords in variation simplex codes over $\mathbb{F}_{p}$ are determined.


\begin{thebibliography}{1}

\bibitem{CB1}
Y. Cassuto, M. Blaum, ``Codes for symbol-pair read channels,"  Proc. IEEE Int. Symp. Inf. Theory(ISIT), pp. 988-992, 2010

\bibitem{CB2}
Y. Cassuto, M. Blaum, ``Codes for symbol-pair read channels," IEEE
Trans. Inf. Theory, vol. 57, no. 12, pp. 8011-8020, Dec. 2011

\bibitem{CL}
Y. Cassuto, S. Litsyn, ``Symbol-pair codes: algebraic constructions
and asymptotic bounds," Proc. IEEE Int. Symp. Inf.
Theory(ISIT),   pp. 2348-2352, 2011

\bibitem{CJKWY}
Y.M. Chee, L. Ji, H.M. Kiah, C. Wang, J. Yin, ``Maximum distance separable codes for symbol-pair read channels," IEEE Trans. Inf. Theory, vol. 59, no. 11, pp. 7259-7267, Nov. 2013

\bibitem{CKW}
Y.M. Chee, H.M. Kiah, C. Wang, ``Maximum distance separable symbol-pair codes,"  Proc. IEEE Int. Symp. Inf. Theory(ISIT), pp. 2886-2890, 2012

\bibitem{CLL}
B. Chen, L. Lin, H. Liu, ``Constacyclic symbol-pair codes: Lower
bounds and optimal constructions," IEEE Trans. Inf. Theory, vol. 63,
no. 12, pp. 7661-7666, Dec. 2017

\bibitem{DGZZZ}
B. Ding, G. Ge, J. Zhang, T. Zhang, Y. Zhang, ``New constructions of
MDS symbol-pair codes," Des. Codes Cryptogr., vol. 86, no. 4, pp. 841-859, 2018

\bibitem{DZG}
B. Ding, T. Zhang, G. Ge, ``Maximum distance separable
codes for $b$-symbol read channels," Finite Fields Appl., vol. 49, pp. 180-197, Jan. 2018

\bibitem{DNSS2}
H.Q. Dinh, B.T. Nguyen, A.K. Singh, S. Sriboonchitta, ``Hamming and symbol-pair distance of repeated-root constacyclic codes of prime power lengths over $\mathbb{F}_{p^{m}}+u\mathbb{F}_{p^{m}}$," IEEE Commun. Letters, vol. 22, no. 12, pp. 2400-2403, Dec. 2018

\bibitem{DNSS1}
H.Q. Dinh, B.T. Nguyen, A.K. Singh, S. Sriboonchitta, ``On the
symbol-pair distance of repeated-root constacyclic codes of prime
power lengths," IEEE Trans. Inform. Theory, vol. 64, no. 4, pp. 2417-2430, April 2018

\bibitem{DWLS}
H.Q. Dinh, X. Wang, H. Liu, S. Sriboonchitta, ``On the
symbol-pair distances of repeated-root constacyclic codes of length $2p^s$," Discrete Math., vol. 342, no. 11, pp. 3062-3078, 2019

\bibitem{HTM}
M. Hirotomo, M. Takita, M. Morii, ``Syndrome decoding of symbolpair
codes,"  Proc. IEEE Int. Symp. Inf. Theory(ISIT), pp. 162-166, 2014

\bibitem{HMH}
S. Horii, T. Matsushima, S. Hirasawa, ``Linear programming decoding of binary linear codes for symbol-pair read channels,"  Proc. IEEE Int. Symp. Inf. Theory(ISIT), pp. 1944-1948, 2016

\bibitem{KZL}
X. Kai, S. Zhu, P. Li, ``A Construction of New MDS Symbol-Pair
Codes," IEEE Trans. Inf. Theory, vol. 61, no. 11, pp. 5828-5834, Nov. 2015

\bibitem{KZZLC}
X. Kai, S. Zhu, Y. Zhao, H. Luo, Z. Chen, ``New MDS Symbol-Pair
Codes from repeated-root codes," IEEE Commun. Letters, vol. 22, no. 3, pp. 462-465,  2018

\bibitem{LG}
S. Li, G. Ge,``Constructions of maximum distance separable
symbol-pair codes using cyclic and constacyclic codes," Designs, Codes
Cryptogr., vol. 84, no. 3, pp. 359-372, 2017

\bibitem{LXY}
S. Liu, C. Xing, C. Yuan, ``List Decodability of Symbol-Pair Codes," IEEE Trans. Information Theory, vol. 65, no. 8, pp. 4815-4821, Aug. 2019

\bibitem{MS}
F.J. MacWilliams, N.J.A. Sloane, ``The theory of error-correcting codes," The Netherlands, North. Holland, Amsterdam, 1977

\bibitem{MHT}
M. Morii, M. Hirotomo, M. Takita, ``Error-trapping decoding for cyclic codes over symbol-pair read channels," Proc. IEEE Int. Symp. Inf. Theory(ISIT), pp. 681-685, 2016

\bibitem{YBS2}
E. Yaakobi, J. Bruck, P.H. Siegel, ``Constructions and decoding of cyclic codes over b-symbol read channels," IEEE Trans. Inf. Theory, vol. 62, no. 4, pp. 1541-1551, April 2016

\bibitem{YBS1}
E. Yaakobi, J. Bruck, P.H. Siegel, ``Decoding of cyclic codes over
symbol-pair read channels," Proc. IEEE Int. Symp. Inf. Theory(ISIT),
 pp. 2891-2895, 2012

\bibitem{SZW}
Z. Sun, S. Zhu, L. Wang, ``The symbol-pair distance distribution of
a class of repeated-root cyclic codes over $\mathbb{F}_{p^{m}}$," Cryptogr. Commun., vol. 10, no. 4, pp. 643-653, 2018

\end{thebibliography}
\end{document}